\documentclass[12pt]{amsart}

\usepackage{amssymb,amsthm,amsmath}
\usepackage[numbers,sort&compress]{natbib}
\usepackage{color}
\usepackage{graphicx}
\usepackage{tikz}
\usepackage{amssymb,amsthm,amsmath}
\usepackage[numbers,sort&compress]{natbib}
\usepackage{color}
\usepackage{graphicx}
\usepackage{tikz}
\usepackage{xparse}
\usepackage{mathrsfs}

\title[ ]{ Irreducibility of the Fermi variety for discrete periodic Schr\"odinger operators and embedded eigenvalues}

\author{Wencai Liu}
\address[W. Liu]{ Department of Mathematics, Texas A\&M University, College Station, TX 77843-3368, USA} \email{liuwencai1226@gmail.com; wencail@tamu.edu}



\theoremstyle{plain}
\newtheorem{theorem}{Theorem}[section]
\newcommand{\R}{\mathbb{R}}
\newtheorem{corollary}[theorem]{Corollary}
\newtheorem{lemma}[theorem]{Lemma}
\newtheorem{proposition}[theorem]{Proposition}
\newtheorem{remark}{Remark}
\newcommand{\C}{\mathbb{C}}
\newcommand{\T}{\mathbb{T}}
\newcommand{\Z}{\mathbb{Z}}

\theoremstyle{plain}
\newtheorem{definition}{Definition}
\newtheorem{conjecture}{Conjecture}

\begin{document}
	
	
	\begin{abstract}
		Let $H_0$ be a discrete periodic  Schr\"odinger operator on $\ell^2(\Z^d)$:
		$$H_0=-\Delta+V,$$
		where $\Delta$ is the discrete Laplacian and $V:\Z^d\to \C$ is periodic.    We prove that  for any $d\geq3$,    the Fermi variety at every energy level  is irreducible  (modulo periodicity).  For $d=2$,    we prove that the Fermi variety at every energy level except for the average of  the potential    is irreducible  (modulo periodicity) and  the Fermi variety at the average of  the potential has at most two irreducible components  (modulo periodicity). 
 This is sharp since for  $d=2$ and a constant potential  $V$,   
		the Fermi variety at  $V$-level  has exactly  two irreducible components (modulo periodicity).  
		We also prove that  the Bloch variety  is irreducible 
		(modulo periodicity)  for any $d\geq 2$.  
		
		As  applications, we prove that when  $V$ is a  real-valued periodic function, 
		the level set of any  extrema of any  spectral band functions,  spectral band edges  in particular,  
		has dimension at most $d-2$  for any $d\geq 3$,  and  finite  cardinality
		for $d=2$. 
		We also  show  that $H=-\Delta +V+v$ does not have any embedded eigenvalues provided that $v$ decays super-exponentially.  
	\end{abstract}
	
	\maketitle 
	\section{Introduction and main results}
	
	Periodic  elliptic  operators  have been  studied intensively  in both mathematics and physics, in particular for their role in  solid state theory.  One of the  difficult and unsolved   problems    is the (ir)reducibility of  Bloch and  Fermi  varieties 
	~\cite{GKTBook,ktcmh90,bktcm91,bat1,batcmh92,ls,shi2,fls,GKToverview,shva}.
	Besides its  importance in algebraic geometry, the
	(ir)reducibility is crucial  in the study of spectral properties of periodic elliptic operators, e.g., the structure of spectral  band edges and 
	the  existence of embedded eigenvalues  under  a    suitable decaying perturbation  of the potential ~\cite{kv06cmp,kvcpde20,shi1,IM14,AIM16}.
	We refer readers to  a survey ~\cite{ksurvey} for the history and most recent developments.

	In this paper, we will concentrate  on     discrete  periodic  Schr\"odinger operators  on $\Z^d$.
	Given $q_i\in \Z_+$, $i=1,2,\cdots,d$,
	let $\Gamma=q_1\Z\oplus q_2 \Z\oplus\cdots\oplus q_d\Z$.
	We say that a function $V: \Z^d\to \C$ is  $\Gamma$-periodic (or just periodic)  if 
	for any $\gamma\in \Gamma$,  $V(n+\gamma)=V(n)$.

	Let  $\Delta$ be the discrete Laplacian on $\ell^2(\Z^d)$, namely
	\begin{equation*}
	(\Delta u)(n)=\sum_{||n^\prime-n||_1=1}u(n^\prime),
	\end{equation*}
	where $n=(n_1,n_2,\cdots,n_d)\in\Z^d$, $n^\prime=(n_1^\prime,n_2^\prime,\cdots,n_d^\prime)\in\Z^d$ and 
	\begin{equation*}
	||n^\prime-n||_1=\sum_{i=1}^d |n_i-n^\prime_i|.
	\end{equation*}
		We consider the  discrete    Schr\"{o}dinger operator on $\ell^2({\Z}^d)$,
	\begin{equation} \label{h0}
	H_0=-\Delta +V .
	\end{equation}
	
	In this paper,  we always assume 	the greatest common factor of $q_1,q_2,\cdots, q_d$ is 1, $V$ is periodic and
	$H_0$ is the  discrete periodic Schr\"{o}dinger operator given by ~\eqref{h0}.
	
	Let $\{\textbf{e}_j\}$, $j=1,2,\cdots d$,  be the   standard basis in $\Z^d$:
	\begin{equation*}
	\textbf{e}_1 =(1,0,\cdots,0),\textbf{e}_2 =(0,1,0,\cdots,0),\cdots, \textbf{e}_{d}=(0,0,\cdots,0,1).
	\end{equation*}
	
	\begin{definition}
		The   { \it Bloch variety} $B(V)$ of  $-\Delta+V$ consists of all pairs $(k,\lambda )\in \C^{d+1}$ for which
		there exists a non-zero solution of the equation 
		\begin{equation} \label{spect_0}
			(-\Delta u)(n)+V(n) u(n)=\lambda u(n) ,n\in\Z^d,  
		\end{equation}
		satisfying the so called Floquet-Bloch  boundary condition
		\begin{equation} \label{Fl}
		u(n+q_j\textbf{e}_j)=e^{2\pi i k_j}u(n),j=1,2,\cdots,d, \text{ and }  n\in \Z^d, 
		\end{equation}
		where $k=(k_1,k_2,\cdots,k_d)\in \C^d$.
		
	\end{definition}
	\begin{definition}
	Given  $\lambda\in \C$,
		the  Fermi surface (variety) $F_{\lambda}(V)$ is defined as the level set of the  Bloch variety:
		\begin{equation*}
		F_{\lambda}(V)=\{k: (k,\lambda)\in B(V)\}.
		\end{equation*}
	\end{definition}
	Our main interest in the present paper is   the irreducibility of Bloch and Fermi    varieties as analytic  sets.   
	\begin{definition}
		A subset   $\Omega\subset \C^k$ is called an analytic set if for any $x\in \Omega$, there is a neighborhood $U\subset \C^k$ of $x$, and analytic functions 
		$f_1,f_2,\cdots,f_p$ in $U$ such that
		\begin{equation*}
		\Omega\cap U=\{y\in U: f_1(y)=0,f_2(y)=0,\cdots,f_p(y)=0\}.
		\end{equation*}
	\end{definition}
	\begin{definition}
		An analytic set $ \Omega$ is said to be irreducible if it can not be represented as the union of two   non-empty proper  analytic subsets.  
	\end{definition}

	It is widely believed that the 
	Bloch/Fermi  variety (modulo periodicity)  is  always irreducible for periodic Schr\"odinger operators \eqref{h0}, which has been 
	formulated  as   conjectures:
	\begin{conjecture} ~\cite[Conjecture 5.17]{ksurvey} \label{conjb}
		The Bloch variety 
		$B(V)$  is irreducible  (modulo periodicity).
	\end{conjecture}
	\begin{conjecture} ~\cite[Conjecture 5.35]{ksurvey} ~\cite[Conjecture 12]{kvcpde20} \label{conj1dc}
		Let  $d\geq 2$.  Then 
		$F_{\lambda}(V)/\Z^d$  is irreducible, possibly except for  finitely 
		many $\lambda\in \C$.
	\end{conjecture}
We remark that in Conjecture 1, 
the irreducibility of  Bloch variety  modulo periodicity means 
for any two irreducible components $\Omega_1$ and $\Omega_2$  of    $B(V)$, there exists $k\in\Z^d$ such that $\Omega_1=(k,0)+\Omega_2$.
In Conjecture 2,  for fixed $\lambda$, 	$F_{\lambda}(V)/\Z^d$  is irreducible means 
for any two irreducible components $\Omega_1$ and $\Omega_2$  of    $F_{\lambda}(V)$, there exists $k\in\Z^d$ such that $\Omega_1=k+\Omega_2$.


	Conjectures \ref{conjb} and \ref{conj1dc}   have  been mentioned in  many  articles ~\cite{batcmh92,bktcm91,ktcmh90,GKTBook,bat1,kv06cmp}. It  seems  extremely hard to prove them, even for  ``generic"  periodic potentials.
	See  Conjecture 13 in ~\cite{kvcpde20} for a ``generic" version of Conjecture  \ref{conj1dc}.

	In this paper,  we will first  prove  both conjectures. For  any $d\geq3$,  we prove that  the Fermi variety at every level  is irreducible  (modulo periodicity).     For $d=2$,    we prove that  the Fermi variety at every level except for the average of the potential  is irreducible  (modulo periodicity).  We also prove that the Bloch variety  is irreducible  (modulo periodicity)  for any $d\geq 2$.

	\begin{theorem}\label{gcf1}
		Let $d\geq3$.  Then the Fermi variety  $F_{\lambda}(V)/\Z^d$ is irreducible for any $\lambda\in \C$.
		
	\end{theorem}
	Denote by $[V]$ the average of $V$ over one periodicity cell, namely
	\begin{equation*}
	[V]=\frac{1}{q_1q_2\cdots q_d}\sum_{0\leq n_1\leq q_1-1\atop{\cdots\cdots\atop{0\leq n_d\leq q_d-1}}}V(n_1,n_2,\cdots,n_d).
	\end{equation*}
	
	\begin{theorem}\label{thm21}
		Let $d=2$.
		Then the  Fermi variety $F_{\lambda}(V)/\Z^2$  is irreducible for any $\lambda\in \C$ except  maybe  for $\lambda=[V] $. Moreover, if  $F_{[V]}(V)/\Z^2$ is reducible, it  has exactly two irreducible components.
	\end{theorem}

	\begin{theorem}\label{corbv1}
		Let $d\geq2$.   Then the Bloch variety  $B(V)$ is irreducible (modulo periodicity).
	\end{theorem}
\begin{remark}
	\begin{enumerate}
		\item The special situation with  the Fermi variety at the  average level in Theorem \ref{thm21}  is not surprising. 
		When $d=2$,
		for  a constant function  $V$,   $F_{[V]}(V)/\Z^2$   has  two irreducible components.  
		\item We should mention that in Theorems \ref{gcf1},  \ref{thm21}, and   \ref{corbv1},  $V$ is allowed to be any complex-valued periodic function.  
			\item 	It  is easy to show  that   Conjecture \ref{conjb}  holds for $d=1$.  See p.18 in ~\cite{GKTBook} for a proof.
	\end{enumerate}

\end{remark}

	Significant progress in proving  those  Conjectures  
	 has been made for $d=2,3$.
	When  $d=2$,  Theorem \ref{corbv1}  was  proved by B{\"a}ttig  ~\cite{battig1988toroidal}.
	In ~\cite{GKTBook},  Gieseker, Kn\"orrer and Trubowitz proved  that 
	$F_{\lambda}(V)/\Z^2$ is irreducible except for finitely many values of  $\lambda$.
	When   $d=3$,  Theorem \ref{gcf1}   has been proved by B{\"a}ttig  ~\cite{batcmh92}.

	For continuous (rather than discrete)  periodic Schr\"odinger operators,  Kn\"orrer and Trubowitz proved  that  the Bloch variety    is irreducible  (modulo periodicity)
	when $d=2$  ~\cite{ktcmh90}.
	
	When the periodic  potential  is  separable,  B{\"a}ttig,  Kn\"orrer and Trubowitz proved that 
	the Fermi variety  at any level  is irreducible   (modulo periodicity) 
	for $d=3$  ~\cite{bktcm91}. 
	
	
	In  ~\cite{GKTBook,ktcmh90,bktcm91,bat1,batcmh92,battig1988toroidal},   proofs   heavily  depend on the construction of toroidal   and directional  compactifications of Fermi  and  Bloch varieties.
	
	A   new approach will be introduced in this  paper. 
	Instead of compactifications, we 
	focus on studying  the  Laurent polynomial  $\mathcal{P}$  arising from the eigen-equation \eqref{spect_0} and \eqref{Fl} after changing the variables.   
	We develop an approach to study the irreducibility of a class of    Laurent polynomials. 
	Firstly, we  show that the closure of the  zero set  of every factor of the   Laurent  polynomial   $\mathcal{P}$ must contain either $z_1=z_2=\cdots =z_d=0$ or $z_1=z_2=\cdots =z_{d-1}=0,z_d=\infty$.   Secondly, we prove that ``asymptotics" of the   Laurent  polynomial at  $z_1=z_2=\cdots =z_d=0$ and  $z_1=z_2=\cdots =z_{d-1}=0, z_d=\infty$ are irreducible. This allows us to conclude  that the   Laurent   polynomial  $\mathcal{P}$  has at most two non-trivial factors.
	Finally, we use degree arguments to show that the only case that $\mathcal{P}$  has    two factors is $d=2$ and $\lambda=[V]$, which completes the proof.   We mention that the irreducibility of  the  Laurent  polynomial allows  a difference of  monomials (see Def.\;\ref{de1}), same issue applies to the  calculations of ``asymptotics". This creates an extra difficulty  in the degree arguments. 
	We introduce  a   polynomial $\mathcal{P}_1$ based on the Laurent polynomial $\mathcal{P}$  multiplying   by a proper monomial.  Delicately playing between the polynomial $\mathcal{P}_1$   and the Laurent polynomial $\mathcal{P}$ is another significant ingredient 	to make the whole proof work.
	
	Although the proof is  written for   Laurent  polynomials coming from the Fermi variety of discrete periodic Schr\"odinger  operators, 
	it  works for a larger class of     Laurent  polynomials. Some ideas  developed in the proof    have been  extended to study 
	  the irreducibility of the Bloch variety in  more general settings ~\cite{flm}.

Irreducibility is a powerful tool to study many aspects of  the  spectral theory of  periodic operators. 
	Let $Q=q_1q_2\cdots q_d$. Assume  that $V$ is a real valued periodic potential. Thus $H_0=-\Delta+V$ is a self-adjoint operator on $\ell^2(\Z^d)$ and its 
	 spectrum 
	\begin{equation}\label{gband}
	\sigma (H_0)={\bigcup }_{m= 1}^{Q}[a_m,b_m]
	\end{equation}
	is the union of spectral bands $[a_m,b_m]$, $m=1,2,\cdots, Q$,  which is the range of a band function $\lambda_m (k)$, $k\in\R^d$. See Section \ref{Sdis} for the precise definition of $\lambda_m(k)$. 
	
	The structure of extrema of band functions plays a  significant role  in many problems, such as  homogenization theory, Green's function asymptotics and Liouville type theorems. We refer readers to ~\cite{kusu01,ksurvey,fk18,col91,dks} and references therein for more details.
	
	It is well known and widely believed  that generically the band functions are Morse functions.  The following conjecture gives a precise description.
	\begin{conjecture}~\cite[Conjecture 5.25]{ksurvey}~\cite[Conjecture 5.1]{kp07}~\cite[Conjecture 5]{dks}\label{cband}
		Generically (with respect to the potentials and other free parameters of the operator), the extrema of band functions
		\begin{itemize}
			\item[(1)]are attained by a single band;
			\item[(2)] are isolated;
			\item[(3)]are nondegenerate, i.e., have nondegenerate Hessians.
		\end{itemize}

	\end{conjecture}
The statement (1) of Conjecture   \ref{cband} was proved in ~\cite{kr}.
	Some progress has been made towards   Conjecture \ref{cband}  at the bottom of the spectrum ~\cite{ks87} or small potentials ~\cite{col91}. 
	Recently, a celebrated work of  Filonov and Kachkovskiy ~\cite{fk18} proves that for a  wide class (not ``generic") of 2D periodic elliptic operators (continuous version), the global extrema of all spectral band functions are isolated.

	 
		
	
	As an application of  the  irreducibility\footnote{Indeed,  a  much weaker assumption is sufficient for our arguments. See Remark \ref{relast}.} (Theorem \ref{thm21}) and Theorem \ref{thmextrem} in Section \ref{Smain}, we  are able to prove  a  stronger version (work for all extrema) of Filonov and Kachkovskiy's results ~\cite{fk18}
	in  the discrete settings.
	The advantage for discrete cases is that the Fermi variety is   algebraic in Floquet variables $ e^{2\pi ik_j}$, $j=1,2,\cdots,d$ which allows us to  use B\'ezout's theorem to do the proof.

	\begin{theorem}\label{thmex2}
		Let $d=2$.  Let $\lambda_{*}$ be an extremum  of   $\lambda_m(k)$, $k\in [0,1)^2,  m=1,2,\cdots,Q$. Then the level set
		\begin{equation}\label{last18}
		\{k\in[0,1)^2:  \lambda_{{m}}(k)=\lambda_{*} \}
		\end{equation}
		has cardinality at most $4(q_1+q_2)^2$.
	\end{theorem}
	In particular, Theorem \ref{thmex2} shows that any  extremum  of any band function can only be attained at  finitely many   points, which  is   a stronger version (not ``generic") than  the statement  (2) of Conjecture \ref{cband}.

	It is worth  pointing out that  Theorem \ref{thmex2} may not hold   for   discrete periodic Schr\"odinger operators  on  a diatomic lattice in $\Z^2$ ~\cite{fk18}.
	\begin{theorem}\label{thmex3}
		Let $d\geq 3$.  Let $\lambda_{*}$ be an extremum of  $\lambda_m(k)$, $k\in [0,1)^d, m=1,2,\cdots,Q$. Then the level set
		\begin{equation*}
		\{k\in[0,1)^d:  \lambda_{{m}}(k)=\lambda_{*}\}
		\end{equation*}
		has dimension at most $d-2$.
	\end{theorem}
	Since  the edge of each spectral band  is an extremum of the band function,  immediately we have  the following two corollaries.
	\begin{corollary}\label{corband2}
		Let $d=2$.  
		Then both  level sets
		\begin{equation*}
		\{k\in[0,1)^2:  \lambda_{{m}}(k)=a_m \} 
		\text{ and }
		\{k\in[0,1)^2:  \lambda_{{m}}(k)=b_m\} 
		\end{equation*}
		have  cardinality at most $4(q_1+q_2)^2$.
	\end{corollary}
	
	\begin{corollary}\label{corband3}
		Let $d\geq 3$.  
		Then both level sets
		\begin{equation*}
		\{k\in[0,1)^d:  \lambda_{{m}}(k)=a_m \} 
		\text{ and }
		\{k\in[0,1)^d:  \lambda_{{m}}(k)=b_m \} 
		\end{equation*}
		have dimension at most $d-2$.
	\end{corollary}
	
	\begin{remark}
The   statements  in Theorem \ref{thmex3} and Corollary \ref{corband3} are sharp for periodic Schr\"odinger operators on a particular lattice  in $\Z^d$ \cite{par1}.
	\end{remark}
	
	The  results of   Corollary \ref{corband2} without the explicit bound of the cardinality and Corollary   \ref{corband3}  were announced by I. Kachkovskiy ~\cite{fk5} during a seminar talk at TAMU, as a part of a joint work with N. Filonov ~\cite{fk2}.  During   Kachkovskiy's talk, we realized that  we could  provide the    approach to study   the upper bound of    dimensions  of    level sets of extrema based on    the Fermi variety. 
	In private communication, we were made aware that the proof from ~\cite{fk2} extends to   Theorem \ref{thmex2} without the explicit bound of the cardinality  and Theorem \ref{thmex3}. However, their approach is very different and is based on the arguments from ~\cite{fk18}.


	We are going to  talk about another application.
	Let
	us introduce  a perturbed periodic operator:
	\begin{equation}\label{gh}
	H=H_0+v=-\Delta +V +v,   
	\end{equation}
	where $v:\Z^d\to \C$ is a decaying function.

	The (ir)reducibility of the Fermi variety is   closely related  to  the existence of eigenvalues embedded into    spectral bands  of perturbed periodic operators ~\cite{kvcpde20,kv06cmp}. 
	We postpone the full set up and background to Section \ref{Smain}, and formulate one main theorem before closing this section.
	Based on the irreducibility (Theorems \ref{gcf1} and \ref{thm21}), the arguments in ~\cite{kvcpde20}, 
	and a unique continuation result for the discrete Laplacian on $\Z^d$,  we are able to prove that
	\begin{theorem}\label{thm1new}
		Assume that $V$ is real and periodic.
		If  
		there exist constants $C>0$ and $\gamma>1$ such that the complex-valued function $v :\Z^d\to\C$ satisfies 
		\begin{equation}\label{ggdecaynew}
		|v(n)|\leq Ce^{-|n|^{\gamma}},
		\end{equation}
		then   $H=-\Delta+V+v$ does not have any embedded eigenvalues, i.e., for any $\lambda \in \bigcup_{m=1}^{Q}(a_m,b_m)$, $\lambda$ is not an eigenvalue of $H$.
	\end{theorem}

Finally, we mention that 
the  irreducibility results established in this paper   provide   opportunities  to  explore more applications  ~\cite{liu21,liuicmp,liuprivate}.
	
	\section{ Main results }\label{Smain}

	\begin{definition}
		Let  $\C^{\star}=\C\backslash \{0\}$ and $z=(z_1,z_2,\cdots,z_d)$. 
		The Floquet variety  is defined as 
		\begin{equation}\label{gff}
		\mathcal{F}_{\lambda}(V)=\{z \in (\C^\star)^d  :z_j=e^{2\pi i k_j}, j=1,2,\cdots,d,  k\in F_{\lambda}(V)\}.
		\end{equation}

	\end{definition}
	\noindent In other words,   $z\in (\C^\star)^d \in \mathcal{F}_{\lambda}(V)$
	if the equation
	\begin{equation}\label{flo}
	 (-\Delta u) (n)+V(n)u(n)=\lambda u (n), n\in\Z
	\end{equation}
	with the boundary condition
	\begin{equation}\label{flo2}
	u(n+q_j\textbf{e}_j)=z_j u(n), j=1,2,\cdots, d, \text{ and } n\in \Z^d
	\end{equation}
has  a   non-trivial   function.
	Introduce a fundamental domain $W$ for $\Gamma$:
	\begin{equation*}
	W=\{n=(n_1,n_2,\cdots,n_d)\in\Z^d: 0\leq n_j\leq q_{j}-1, j=1,2,\cdots, d\}.
	\end{equation*}
	By writing out $-\Delta +V$  as acting on the  $Q$ dimensional space $\{u(n),n\in W\}$, 
	the eigen-equation \eqref{flo} and \eqref{flo2} (\eqref{spect_0} and \eqref{Fl})
	translates into the eigenvalue problem for a  $Q\times Q$ matrix $\mathcal{D}(z)$ ($ D(k)$). Let  
	$\mathcal{P}(z,\lambda)$ (${P}(k,\lambda)$) be the determinant of 
 $ \mathcal{D}(z)-\lambda I$ (${D}(k)-\lambda I$).     We should mention that $\mathcal{D}(z)$ (${D}(k)$) and 	$\mathcal{P}(z,\lambda)$ ($P(k,\lambda)$) depend on the potential $V$. Since the potential is fixed, we drop the dependence  during the proof.
	
	From the notations above, one has that 
		\begin{equation}\label{g11}
F_{\lambda}(V) =\{k\in\C^d: {P}(k,\lambda) =0\}, \mathcal{F}_{\lambda}(V) =\{ z\in (\C^\star)^d : \mathcal{P}(z,\lambda) =0\}.
		\end{equation}
	
It is easy to see that 
	$\mathcal{P}(z,\lambda) $ is a   polynomial  in the variables $\lambda$ and 
	$$z_1,z_1^{-1}, z_2,z_2^{-1}, \cdots,z_d,z_d^{-1}.$$ 
	In other words $\mathcal{P}(z,\lambda) $ is a  Laurent polynomial  of 
	$z_1, z_2, \cdots,z_d$  and polynomial in $\lambda$.
	Therefore, the Floquet variety $\mathcal{F}_{\lambda}(V) $     is an algebraic set\footnote{ Usually, an algebraic set is defined as  common zeros of a collection of polynomials. Here, we  call   $X\subset (\C^{\star})^d$ an algebraic set even though $X$ is the  zeros of a Laurent polynomial.  }. 
	It implies  that both $B(V)$    and $F_{\lambda}(V)$  are  (principal) analytic sets.
	 Since  the identity \eqref{Fl} is unchanged  under the shift: $k\to k+\Z^d$,  it is natural to study  $F_{\lambda}(V)/\Z^d$.

	In our proof, we  focus on studying   the Floquet variety $\mathcal{F}_{\lambda}(V)$ to  benefit from its algebraicity.
	
	A    Laurent polynomial  of a  single term  is called monomial, i.e.,   $Cz_1^{a_1}z_2^{a_2}\cdots z_{k}^{a_k}$, where  $a_j\in\Z$, $j=1,2,\cdots,k$,
	and $C$ is a non-zero constant.
	\begin{definition}\label{de1}
		We say  that  a Laurent polynomial $h\left(z_{1}, z_{2},\cdots,z_k\right)$ is irreducible if it  can not be factorized non-trivially, that is, there are no   non-monomial Laurent polynomials $f\left(z_{1}, z_{2},\cdots,z_k\right)$  and $g\left(z_{1}, z_{2},\cdots,z_k\right)$ such that 
		$h=fg$.
	\end{definition}
\begin{remark}
When $h$ is a polynomial, the definition of irreducibility in Def. \ref{de1}  differs  the traditional one\footnote{A  polynomial $h$  is called irreducible if  there are no  non-constant polynomials $f$ and $g$ such that $ h=fg$.} (because of the monomial).   For example, the polynomial $z^2+z$ is irreducible according to  Def. \ref{de1}.  
 This will not create any trouble since all  polynomials arising from this paper do not have factors $z_j$, $j=1,2,\cdots,k$. 
\end{remark}

	Based on the above notations and definitions, we have the following simple facts.
		\begin{proposition}\label{prop}
		Fix $\lambda\in \C$. We have
		\begin{enumerate}
			\item The Fermi variety/surface  ${F}_{\lambda}(V)/\Z^d$  is irreducible if and only if $\mathcal{F}_{\lambda}(V)$  is 
		irreducible;
		\item 	If the  Laurent polynomial  	$\mathcal{P}(z,\lambda)$ (as a function of $z$) is irreducible, then 
		$\mathcal{F}_{\lambda}(V)$  is irreducible.
		\end{enumerate}
	\end{proposition}
	
%

	\begin{theorem}\label{gcf}
		Let $d\geq3$.   Then for  any $\lambda\in \C$,  the Laurent polynomial $\mathcal{P}(z,\lambda) $ (as a function of $z$) is irreducible. 
		
	\end{theorem}

	\begin{theorem}\label{thm2}
		Let $d=2$.
		Then 	the Laurent polynomial   $ \mathcal{P}(z,\lambda)$ (as a function of $z$)  is irreducible for any $\lambda\in \C$ except maybe for  $\lambda=[V] $, where $[V]$ is the average of $V$ over one periodicity cell. Moreover, if  $ \mathcal{P}(z,[V])$ is reducible,   
		$ \mathcal{P}(z,[V])$ has exactly two distinct non-trivial  irreducible factors (each factor has  multiplicity one).
	\end{theorem}
By Theorems \ref{gcf} and \ref{thm2},  and some basic properties of $\mathcal{P}$, 	 we immediately obtain 
		\begin{theorem}\label{thmbv1ne}
	Let $d\geq2$.  Then the  Laurent polynomial   $\mathcal{P}(z,\lambda) $ (as a function of both $z$  and $\lambda$) is irreducible. 
	\end{theorem}

	\begin{remark}
		\begin{enumerate}
			\item 	By \eqref{g11} and Prop.\ref{prop},  Theorems \ref{gcf1}, \ref{thm21} and \ref{corbv1}   follow from	Theorems \ref{gcf}, \ref{thm2} and \ref{thmbv1ne}.
		\item Denote by 	${\bf 0}$  the zero $\Gamma$-periodic potential. From \eqref{last} below, one can see that if $ \mathcal{P}(z,[V])$ is reducible ($d=2$),  then    $F_{[V]} (V)=F_{0}({\bf 0})$.
		\end{enumerate}

\end{remark}

	\begin{remark}
		Reducible Fermi surfaces are known to occur for  periodic graph operators, even at all energy levels, e.g.,  ~\cite{shi1,fls}. 
	\end{remark}

	Our next topic is about the extrema of  band functions.  

	\begin{theorem}\label{thmextrem}
		Assume that $V$ is a real valued periodic potential. 
		Let $\lambda_{*}$ be an extremum of a band function $\lambda_m(k)$, for some $m=1,2,\cdots,Q$.
		Then we have 
		\begin{equation}\label{gextrem}
		\{k\in \R^d: \lambda_m(k)=\lambda_{*}\} \subset \{k\in\R^d: {P}(k,\lambda_{*})=0, |\nabla_k {P}(k,\lambda_{*})|=0\},
		\end{equation}
		where $\nabla$ is the gradient.
	\end{theorem}

Recall that a point $x$ of an analytic set $\Omega$ 
is called a regular point if there is a neighborhood 
$U$ of $x$
such that 
$U\cap \Omega$
is an   analytic manifold. Any other point is called a singular point.

By  Theorems \ref{gcf} and  \ref{thm2}, one has that for any fixed $\lambda$, 
$\mathcal{P}(z,\lambda)$ (${P}(k,\lambda)$) is  a minimal defining function  (see p.27 in \cite{chi} for the precise definition)  of  $\mathcal{F}_\lambda (V)$ ($F_{\lambda}(V)$).
Therefore, Theorem \ref{thmextrem} implies (see p.27 in \cite{chi})
\begin{corollary}
			Let $\lambda_{*}$ be an extremum of a band function $\lambda_m(k)$, $k\in\R^d$, for some $m=1,2,\cdots,Q$.
			Then  $
			\{k\in \R^d: \lambda_m(k)=\lambda_{*}\} 
		$ is a subset of  singular points of the  Fermi variety $F_{\lambda_*}(V)$.
\end{corollary}

	The last topic we are going to  discuss is the  existence of  embedded eigenvalues  for  perturbed discrete  periodic operators
	\eqref{gh}.

	For $d=1$, the  existence/absence  of  embedded eigenvalues   has been understood very well ~\cite{ld,krs,ns12,kn07,rof64,liucriteria}.
	Problems of the existence of embedded eigenvalues  in  higher dimensions   are  a lot more complicated. The techniques of the  generalized Pr\"ufer transformation and oscillated integrals developed  for $d=1$  are not available.
	

	In ~\cite{kvcpde20}, Kuchment and Vainberg introduced  a    new approach to study the embedded  eigenvalue problem  for perturbed    periodic operators.
	It employs the analytic structure of the Fermi variety, unique continuation results, and techniques of several complex  variables theory.

	{\bf Condition 1:}  Given  $\lambda\in\bigcup(a_m,b_m)$, we say that  $\lambda$   satisfies Condition 1 if   any irreducible component of the   Fermi variety $F_{\lambda} (V)$ contains an open   analytic hypersurface   of dimension $d-1$ in $\R^d$.

	\begin{theorem}\label{thm1kv}~\cite{kvcpde20}
		Let $d=2,3$,  and $H_0$ and $H$  be  continuous versions of  \eqref{h0} and \eqref{gh} respectively.
		Assume that 	there exist constants $C>0$ and $\gamma>4/3$ such that
		\begin{equation*}
		|v(x)|\leq Ce^{-|x|^{\gamma}}.
		\end{equation*}
		Assume  Condition 1 for some  $\lambda\in \bigcup (a_m, b_m)$.
		Then this $\lambda $ can not be an eigenvalue of $H=-\Delta+V+v$.
	\end{theorem}
	
	For $\lambda$ in the interior of a spectral band, 	the irreducibility of the Fermi variety  $F_{\lambda}(V)$   implies   Condition 1 for this $\lambda$. See Lemma   \ref{led-1}.
	The restriction on $d=2,3$  
	and the critical exponent  $4/3$ arise  from a quantitative unique continuation result. 
	Suppose $u$ is a solution of 
	\begin{equation*}
	-\Delta u+\tilde{V}u=0 \text{ in } \R^d,
	\end{equation*}
	where    $|\tilde{V}|\leq C$, $|u|\leq C$ and $u(0)=1$.
	From the unique continuation
	principle,  $u$ cannot vanish identically on any open set.
	The quantitative result states ~\cite{BK05}
	\begin{equation}\label{gcon43}
	\inf_{|x_0|=R}  \sup_{|x-x_0|\leq 1}|u(x)|  \geq  e^{-CR^{4 / 3} \log R}.
	\end{equation}

		A   similar   version of \eqref{gcon43} was established  in ~\cite{Mesh92}  (also see Remark 2.6 in ~\cite{fhh}), namely, 
	there is  no non-trivial  solution of $(-\Delta +\tilde{V})u=0$ such that 
	\begin{equation}\label{gcon43w}
		|u(x)| \leq  e^{-C|x|^{4 / 3}} \text{ for any } C>0.
	\end{equation}


	For complex potentials $\tilde{V}$,
	the critical exponent  $ 4/3$  in  \eqref{gcon43} and \eqref{gcon43w} is optimal  in view of the Meshkov's
	example ~\cite{Mesh92}. It  has been conjectured (referred to as Landis' conjecture, which is still open for $d\geq3$) that the critical exponent is 1 for real potentials.  See ~\cite{ksw15,dkw19,lmnn} and references therein for the recent progress of the Landis' conjecture.
	However, the   unique  continuation principle for discrete Laplacians is well known not to  hold (see e.g., ~\cite{ku911,Jf07}).  
	This issue turns out to be the   obstruction to generalize Kuchment-Vainberg's approach to discrete 
	periodic  Schr\"odinger operators ~\cite{private}. 
	
	Fortunately, we realize that   a    weak  unique continuation result  is sufficient   for Kuchment-Vainberg's arguments 
	in  ~\cite{kvcpde20}.   Such a  unique continuation result is not difficult to    establish  for discrete Schr\"odinger operators on $\Z^d$. 
	Actually, the critical  component  can be  improved from ``$4/3$" to ``1". Therefore,
	we are able to establish the discrete version of Theorem \ref{thm1kv} for any dimension.

	\begin{theorem}\label{thm1}
		Assume $V$ is a real valued periodic function. 
		Let $d\geq 2$, $H_0$ and  $H$ be given by \eqref{h0} and \eqref{gh} respectively.
		Assume that 
		there exist constants $C>0$ and $\gamma>1$ such that
		\begin{equation}\label{ggdecay}
		|v(n)|\leq Ce^{-|n|^{\gamma}}.
		\end{equation}
		Assume  Condition 1  for some $\lambda\in \bigcup_{m=1}^{Q}(a_m,b_m)$.
		Then this $\lambda $ can not be an eigenvalue of $H=-\Delta+V+v$.
	\end{theorem}
	
	\begin{remark}
		
		\begin{itemize}
			\item It is well known that for general periodic graphs even compactly supported solutions can exist (see e.g.~\cite{ku911}).
			\item  It is known that a  compactly supported  perturbation  of the operator on a graph might have  an embedded eigenvalue.
			If this case happens,  under  the assumption on    irreducibility of  the Fermi variety, 
			Kuchment and Vainberg  proved that  the corresponding eigenfunction is compactly supported (invalid the unique continuation) ~\cite{kv06cmp}. Shipman provided  examples of periodic graph operators with  unbounded support eigenfunctions 
			for embedded eigenvalues (the Fermi variety is reducible at every energy level) ~\cite{shi1}. 
			
		\end{itemize}
		
	\end{remark}

	
	Assume that $V$ is zero, which  can be viewed as  a  $\Gamma$-periodic function for any  $\Gamma$. 
	Denote by 
	$[a_m,b_m]$, $m=1,2,\cdots, Q$,  the  spectral bands  of $-\Delta$. Clearly,
	$${\bigcup }_{m = 1}^{Q}[a_m,b_m]=\sigma(-\Delta)=[-2d,2d].$$
	
	\begin{lemma}~\cite[Lemmas 1.2 and 1.3]{hj18}\label{Enot0}
		Let $d\geq 2$.
		Then
		\begin{itemize}
			\item for any  $\lambda\in (-2d, 2d)\setminus \{0\}$,  $\lambda \in  (a_m,b_m)$ for some $1\leq m\leq Q$,
			\item  	if  at least one of $q_j$'s is odd, then $0\in (a_m,b_m)$ for some $1\leq m \leq Q$.
		\end{itemize}
	\end{lemma}
	
	For $d=2$, Lemma  \ref{Enot0}    was  also  proved in ~\cite{ef}.  Based on Lemma \ref{Enot0},  Han and Jitomirskaya   proved 
	the discrete Bethe-Sommerfeld conjecture \cite{hj18}. See   
	~\cite{par,dtcmh} for  the continuous Bethe-Sommerfeld conjecture.
	
	Theorem \ref{thm1new} and Lemma \ref{Enot0} imply 
	\begin{corollary}\label{cor14}
		Assume that there exist some $C>0$ and $\gamma>1$ such that 
		\begin{equation*}
		|v(n)|\leq Ce^{-|n|^{\gamma}}.
		\end{equation*}
		Then 
		$\sigma_{ p}(-\Delta+v)\cap(-2d,2d)=\emptyset$.
	\end{corollary}
	\begin{remark}
		Under a stronger  assumption that $v$ has compact support,   Isozaki and Morioka  proved that $\sigma_{ p}(-\Delta+v)\cap(-2d,2d)=\emptyset$ ~\cite{IM14}.
	\end{remark}
	The rest of this paper is organized as follows. The proof of Theorems \ref{gcf},  \ref{thm2} and \ref{thmbv1ne} is entirely self-contained.
	We  recall   the discrete Floquet-Bloch transform in Section \ref{Sdis}. In Section \ref{Sflo},    we do preparations for  proofs.
	Section \ref{Spro} is devoted to proving Theorems    \ref{gcf},   \ref{thm2} and \ref{thmbv1ne}. 
	Sections \ref{Semb} and \ref{sapp} are devoted to proving Theorems  \ref{thm1} and \ref{thmextrem} respectively.
	 In Section \ref{sapp1},  
	we  prove Theorems  \ref{thmex2}, \ref{thmex3}	and \ref{thm1new}.

	\section{ Discrete Floquet-Bloch transform}\label{Sdis}
 In this section, we recall the standard  discrete Floquet-Bloch transform. We refer readers to ~\cite{ksurvey,kru} for details.

Let 
\begin{equation*}
	\bar{W}=\left\{0,\frac{1}{q_1},\frac{2}{q_1},\cdots,\frac{q_1-1}{q_1}\right\}\times\cdots \times \left\{0,\frac{1}{q_d},\frac{2}{q_d},\cdots,\frac{q_d-1}{q_d}\right\}\subset [0,1]^d.
\end{equation*}
Define the discrete Fourier transform $\hat{V}(l) $ for $l\in \bar{W}$ by 
\begin{equation*}
	\hat{V}(l) =\frac{1}{{Q}}\sum_{ n\in {W} } V(n) e^{-2\pi i l\cdot n},
\end{equation*}
where $l\cdot n= \sum_{j=1}^d l_j n_j$ for $l=(l_1,l_2,\cdots,l_d)\in \bar{W}$ and $n=(n_1,n_2,\cdots,n_d)\in \Z^d$.
For convenience, we  extend  $\hat{V}(l)$ to  $\bar{W}+\Z^d$ periodically, namely for any $l\equiv \tilde{l}\mod\Z^d$,
\begin{equation*}
	\hat{V}(l)=\hat{V}(\tilde{l}).
\end{equation*}
The inverse of the discrete Fourier transform  is given by 
\begin{equation*}
	V(n)=\sum_{l \in \bar{W}}\hat{V}(l) e^{2\pi il\cdot n}.
\end{equation*}
For a function $u\in\ell^2(\Z^d)$, its Fourier transform
$\mathscr{F}(u) = \hat{u}:\T^d=\R^d/\Z^d\to\C$ is given by
\begin{equation*}
	\hat{u}(x)=\sum_{n\in \Z^d}u(n)e^{-2\pi i n\cdot x}.
\end{equation*}

For any   periodic function    $V$   and  any $u \in \ell^2(\Z^d)$, one has
\begin{equation*}
	\widehat{Vu}(x)=\sum_{l\in \bar{W} } \hat{V}(l)\hat{u}(x-l).
\end{equation*}
We remark that $\hat{u}$ is the Fourier transform  for $u\in \ell^2(\Z^d)$ and $\hat{V}$ is the discrete Fourier transform for $V(n)$, $n\in{W} $.
Let
\begin{equation*}
	\mathcal{B}=\prod_{j=1}^{d} [0,\frac{1}{q_j}).
\end{equation*}
Let $L^2(  \mathcal{B} \times \bar{W})$ be all functions with  the finite norm given by
\begin{equation*}
	||f||_{L^2( \mathcal{B}\times\bar{ W})}=\sum_{l\in \bar{W}}\int_{\mathcal{B}} |f(x,l)|^2dx.
\end{equation*}
Define the unitary map $U:\ell^2(\Z^d)\to L^2( \mathcal{B}\times \bar{W})$  by 
\begin{equation*}
	(U(u))(x,l)= \hat{u}(x+l)
\end{equation*}
for $x=(x_1,x_2,\cdots,x_d)\in \mathcal{B}$ and $l\in \bar{W}$.
For fixed $x\in \mathcal{B}$, define the operator $\tilde{H}_0(x)$ on $\ell^2(\bar{W})$:
\begin{equation}\label{gh0x}
	(\tilde{H}_0(x) u)(l)=\left(\sum_{j=1}^{d}-2\cos(2\pi(l_j+x_j))u(l)\right)+\sum_{j\in \bar{W}} \hat{V}(l-j)u(j),
\end{equation}
where $l=(l_1,l_2,\cdots,l_d)\in \bar{W}$.
Let  $\hat{H}_0:L^2(\mathcal{B}\times \bar{W})\to L^2(\mathcal{B}\times \bar{W})$  be given by
\begin{equation}\label{last6}
	(\hat{H}_0 u)(x,l)= \left(\sum_{j=1}^{d}-2\cos(2\pi(l_j+x_j))u(x,l)\right)+\sum_{j\in \bar{W}} \hat{V}(l-j)u(x,j).
\end{equation}
The following two Lemmas are well known.
\begin{lemma}
	\label{leuni}
	Let  $H_0=-\Delta+V$.
	Let   $\hat{H}_0$ be given by \eqref{last6}. Then
	\begin{equation}
		\hat{H}_0 = U  H_0  U^{-1}.
	\end{equation}
	
\end{lemma}
\begin{proof}
	Straightforward computations.
\end{proof}
Given $x \in \R^d$, let $\mathscr{F}^{x}$ be the Floquet-Bloch transform  on $ \ell^2(W)$:
for any vector on $W$, $\{u(n)\}_{n\in W}$, 
\[ [\mathscr{F}^{x} u] (n^\prime) = \frac{1}{\sqrt{Q}} \sum_{n \in W} e^{-2\pi i  \sum_{j=1}^d  \left(\frac{n_j^\prime}{q_j}+x_j\right) n_j }u(n), \quad n^\prime \in W .\]

Let $\tilde{{D}} (x)$ be the  $Q\times Q$ matrix given by $D(q_1 x_1,q_2x_2,\cdots,x_d q_d)$. 
\begin{lemma}
	\label{le27}
	The operator $\tilde{H}_0(x)$ given by  \eqref{gh0x} is unitarily equivalent  to  $\tilde{{D}}(x)$. 
\end{lemma}
\begin{proof}
	By \eqref{spect_0} and \eqref{Fl}, 
	$\tilde{{D}}(x)$ is    the restriction of $-\Delta+V$ to $W$ with boundary conditions:
	\begin{equation}\label{flo1}
		u(n+q_j\textbf{e}_j)=e^{2\pi i q_jx_j} u(n),j=1,2,\cdots,d, n\in\Z^d.
	\end{equation}
	Let $T:\ell^2(\bar W)\to \ell^2({W}) $ given by $T(l_1,l_2,\cdots,l_d)=(q_1l_1,q_2l_2,\cdots,q_d l_d)$, where $(l_1,l_2,\cdots,l_d)\in\bar{W}. $
	Direct computations  imply that $$\tilde{H}_0(x)=  T \mathscr{F}^{x} \tilde{{D}}(x) (\mathscr{F}^{x})^* T^{-1}= T \mathscr{F}^{x}\tilde{{D}}(x) (\mathscr{F}^{x})^{-1} T^{-1}.$$
\end{proof}
Assume $V$ is real.
For each $k\in[0,1)^d$, it is easy to see that $D(k)$ has $Q=q_1q_2\cdots q_d$ eigenvalues. Order them in non-decreasing  order
\begin{equation*}
	\lambda_1(k)\leq \lambda_2(k)\leq\cdots \leq \lambda_Q(k).
\end{equation*}
We call   $\lambda_m(k)$ the $m$-th (spectral) band function, $m=1,2,\cdots, Q$.
Then we have
\begin{lemma}\label{lespband} 
	\begin{align*}
		[a_m,b_m]= [\min_{k\in[0,1)^d}\lambda_m (k),\max_{k\in[0,1)^d}\lambda_m (k)]
	\end{align*}
	and $a_m<b_m$, 
	$m=1,2,\cdots, Q$.
	
\end{lemma}


	
	\section{ Preparations }\label{Sflo}
	For readers' convenience, we   collect  some notations   and define  a few new notations  here, which  will be  constantly used in the proofs.
		\begin{enumerate}\label{re8}
			\item  $\mathcal{D}(z)$ is the $Q\times Q$ matrix arising from the eigen-equation \eqref{flo} and \eqref{flo2}.
			
				\item $z_j=e^{2\pi i k_j}$ and $k_j= q_jx_j$, $j=1,2,\cdots,d$.  
			$	\tilde{D}(x) =	D(k) =\mathcal{D}(z)$. $\tilde{\mathcal{D}}(z)=\mathcal{D} (z_1^{q_1},z_2^{q_2},\cdots,z_d^{q_d})$.
			\item 
			$\mathcal{P}(z,\lambda)=\det(\mathcal{D}(z)-\lambda I)$, $\tilde{\mathcal{P}}(z,\lambda)=\det(\tilde{\mathcal{D}}(z)-\lambda I)$, ${P}(k,\lambda)=\det({D}(k)-\lambda I)$,
			$\tilde{P}(x,\lambda)=\det(\tilde{D}(x)-\lambda I)$.
			\item 	Let $$\rho^j_{n_j}=e^{2\pi i \frac{n_j}{q_j}},$$
			where $0\leq n_j \leq q_j-1$, $j=1,2,\cdots,d$. Denote by  $\mu_{q_j}$  the multiplicative group of $q_j$ roots of unity, $j=1,2,\cdots, d$.   Let $\mu=\mu_{q_1}\times \mu_{q_2}\times \cdots \times \mu_{q_d}$.
			
			For any $\rho=(\rho^1,\rho^2,\cdots, \rho^d)\in \mu$, we can define a natural action  on $\C^d   $ 
			$$
			\rho \cdot\left(z_{1}, z_{2},\cdots, z_d\right)=\left(\rho^{1} z_{1}, \rho^{2} z_{2},\cdots,\rho^dz_d\right).
			$$
		
			
			\item 	 For a polynomial $f(z)$, denote by $\deg(f)$ the degree of $f$. 
			\item Let	$\mathcal{P}_1(z,\lambda)= (-1)^Q z_1^{\frac{Q}{q_1} }z_2^{\frac{Q}{q_2}}\cdots z_d^{\frac{Q}{q_d}}\mathcal{P}(z,\lambda)$. 
			\item 	For any $a=(a_1,a_2,\cdots,a_d)\in \Z^d$, let $z^a=z_1^{a_1}z_2^{a_2}\cdots z_d^{a_d}$.
			
		
		\end{enumerate}


	The following lemma is standard. 
	
	\begin{lemma}\label{lesep} 
		Let $n=(n_1,n_2,\cdots,n_d) \in {W}$ and  $n^\prime=(n_1^\prime,n_2^\prime,\cdots,n_d^\prime) \in {W}$. Then 
		$\tilde{\mathcal{D}}(z)$ is unitarily  equivalent to 		
		$
		A+B,
		$
		where $A$ is a diagonal matrix with entries
		\begin{equation}\label{A}
		A(n;n^\prime)=-\left(\left(\sum_{j=1}^d \left(\rho^j_{n_j}z_j+\frac{1}{\rho^j_{n_j} z_j} \right)\right)-\lambda \right) \delta_{n,n^{\prime}}
		\end{equation}
		and $B$
		$$B(n;n^\prime)=\hat{V} \left(\frac{n_1-n_1^\prime}{q_1},\frac{n_2-n_2^\prime}{q_2},\cdots, \frac{n_d-n_d^\prime}{q_d}\right).$$
		In particular,
		\begin{equation*}
		\tilde{\mathcal{P}}(z, \lambda) =\det(A+B).
		\end{equation*}

	\end{lemma}
	\begin{proof}
		Recall that $x_j=\frac{k_j}{q_j}$, $z_j=e^{2\pi i k_j}$, $j=1,2,\cdots,d$. 
		Lemma \ref{lesep}  follows from  Lemma \ref{le27} and \eqref{gh0x}.
	\end{proof}
	We note that  $B$ is  independent of $z_1,z_2,\cdots,z_d$ and $\lambda$.
	
Here are some simple facts about $\mathcal{P}$,   $\tilde{\mathcal{P}} $ and $\mathcal{P}_1$.  
\begin{enumerate}
	\item  $ \mathcal{P}(z,\lambda)$  is symmetric with respect to $z_j$ and $z_j^{-1}$, $j=1,2,\cdots,d$.
	\item   
	$\mathcal{P}(z,\lambda)$  is a   polynomial in the variables  $z_1,z_1^{-1}, z_2,z_2^{-1}, \cdots,z_d,z_d^{-1}$ and   $\lambda$ with highest degree terms (up to a  $\pm$ sign)
	$z_1^{ \frac{Q}{q_1}},z_1^{-\frac{Q}{q_1}}, z_2^{ \frac{Q}{q_2}},z_2^{- \frac{Q}{q_2}} \cdots ,z_d^{ \frac{Q}{q_d}},z_d^{- \frac{Q}{q_d}}$ and $\lambda^{Q}$.
	\item  
	$\tilde{\mathcal{P}}(z,\lambda)$  is a    polynomial  in the variables   $z_1,z_1^{-1}, z_2,z_2^{-1}, \cdots,z_d,z_d^{-1}$ and  $\lambda$ with highest  degree terms (up to a  $\pm$ sign)
	$z_1^{ Q},z_1^{-Q}, z_2^{Q}, z_2^{-Q},\cdots, z_d^{Q},z_d^{-Q}$ and $\lambda^{Q}$.	
	\item 
	$\mathcal{P}_1(z,\lambda)$  is a polynomial of $z$ and $\lambda$.   $\mathcal{P}_1(z,\lambda)$ can not have a factor  $z_j$, $j=1,2,\cdots,d$, namely 
	\begin{equation}\label{nomo}
	z_j\nmid \mathcal{P}_1(z,\lambda), j=1,2,\cdots,d.
	\end{equation}
	Therefore,  the Laurent polynomial  $\mathcal{P}(z,\lambda)$ is irreducible (as a function of $z$) if and only if the polynomial $\mathcal{P}_1(z,\lambda)$ (as a function of $z$) 
	is irreducible in the traditional way, namely, there are no non-constant polynomials $f\left(z\right)$  and $g\left(z\right)$ such that 
	$\mathcal{P}_1(z,\lambda)=f(z)g(z)$. 
	
\end{enumerate}

	\section{Proof of Theorems \ref{gcf} ,  \ref{thm2} and \ref{thmbv1ne}}\label{Spro}


Let 
\begin{equation}\label{cone1'}
	\tilde{h}_1(z)=z_1^Qz_2^Q\cdots z_d^Q\prod_{ 0\leq n_j\leq q_j-1\atop {1\leq j\leq q}} \left(\sum_{j=1}^d\frac{1}{\rho^j_{n_j}z_j}\right),
\end{equation}
and 
\begin{equation}\label{cone22'}
	\tilde{h}_2(z)=z_1^Qz_2^Q\cdots z_{d-1}^Qz_d^{-Q}\prod_{ 0\leq n_j\leq q_j-1\atop {1\leq j\leq q}} \left(\rho^d_{n_d} z_d+ \sum_{j=1}^{d-1}\frac{1}{\rho^j_{n_j}z_j}\right).
\end{equation}
One can see that $	\tilde{h}_1(z) $ is a polynomial in variables $z_1,\cdots,z_{d-1}, z_d$ and $	\tilde{h}_2(z) $ is a polynomial in variables $z_1,\cdots,z_{d-1}, z_d^{-1}$.

Since both $\tilde{h}_1(z)$ and $\tilde{h}_2(z)$ are unchanged under the  action of the   group  $\mu$, we have that there exist $h_1(z)$ (a polynomial of $z_1,\cdots,z_{d-1}, z_d$) and   
$h_2(z)$ (a polynomial of $z_1,\cdots,z_{d-1}, z_d^{-1}$)  such that
\begin{equation}\label{cone3}
	\tilde{h}_1(z_1,z_2,\cdots,z_d)=h_1(z_1^{q_1},z_2^{q_2},\cdots, z_d^{q_d}),
\end{equation}
and 
\begin{equation}\label{cone4}
	\tilde{h}_2(z_1,z_2,\cdots,z_d)=h_2(z_1^{q_1},z_2^{q_2},\cdots, z_d^{q_d}).
\end{equation}

\begin{lemma}\label{key}
	Both $h_1(z)$ and $h_2(z)$ are irreducible.
\end{lemma}
\begin{proof}
	Without loss of generality, we only  show that  $h_1(z)$ is irreducible. 
	Suppose the statement is not true.  Then there are two non-constant  polynomials $f(z)$ and $g(z)$ such that 
	$h_1(z)=f(z)g(z)$.
	Let $$\tilde{f}(z)=f(z_1^{q_1},z_2^{q_2},\cdots, z_d^{q_d}),\tilde{g}(z)=g(z_1^{q_1},z_2^{q_2},\cdots, z_d^{q_d}).$$
	Therefore,
	\begin{equation}\label{key1}
		\tilde{f}(z) \tilde{g}(z)= z_1^Qz_2^Q\cdots z_d^Q	\prod_{ 0\leq n_j\leq q_j-1\atop {1\leq j\leq q}} \left(\sum_{j=1}^d\frac{1}{\rho^j_{n_j}z_j}\right).
	\end{equation}
	By the assumption that the greatest common factor of $q_1,q_2,\cdots,q_d$ is 1,  we have  for any $n_j,n_j^\prime$ with $0\leq n_j,n_j^\prime\leq q_j-1$ and $ (n_1,n_2,\cdots,n_d)\neq (n_1^\prime,n_2^\prime,\cdots, n_d^\prime),$
	\begin{equation}\label{dckey}
		\left\{z\in (\C^{\star})^d: \sum_{j=1}^d\frac{1}{\rho^j_{n_j}z_j}=0  \right\} \neq \left\{z\in (\C^{\star})^d: \sum_{j=1}^d\frac{1}{\rho^j_{n_j^\prime}z_j}=0  \right\} .
	\end{equation}
	By the fact that both  $\tilde{f}(z) $  and $\tilde{g}(z)$   are unchanged under the action $\mu$, and \eqref{dckey}, 
	we  have that  if $\tilde{f}(z) $  (or $\tilde{g}(z)$) has one factor $   \left(\sum_{j=1}^d\frac{1}{\rho^j_{n_j}z_j}\right)$,
	then  $\tilde{f} (z)$   (or $\tilde{g}(z)$)  will have a factor $	\prod_{ 0\leq n_j\leq q_j-1\atop {1\leq j\leq q}} \left(\sum_{j=1}^d\frac{1}{\rho^j_{n_j}z_j}\right)$. This  contradicts  \eqref{key1}.
\end{proof}
\begin{lemma}\label{lesin}
	For any $\lambda\in\C$,	the polynomial	$\mathcal{P}_1(z,\lambda)$ (as a function of  $z$) has at most two  non-trivial  factors (count multiplicity). In the case that  $\mathcal{P}_1 (z,\lambda)$ has two non-trivial factors, namely $\mathcal{P}_1(z,\lambda)=f(z)g(z)$,  we have   that (maybe exchange $f$ and $g$)
	\begin{itemize}
		\item the closure \footnote{The   closure  is taken  in $(\C\cup \{\infty\})^d$.} of   $Z_1=	\{z\in (\C^{\star})^d: f(z)=0\} $ contains $z_1=z_2=\cdots=z_d=0$,

		\item the closure of  $Z_2=	\{z\in (\C^{\star})^d: g(z)=0\} $  contains   $z_1=z_2=\cdots=z_{d-1}=0,z_{d}^{-1}=0$ \footnote{$z_d^{-1}=0$ means $z_d=\infty$.
			In the proof,    we  view $z_d^{-1}$ as a new variable   when $z_d=\infty$. }.
		
	\end{itemize}

\end{lemma}
\begin{proof}
	Let $f(z)$ be a factor of polynomial $\mathcal{P}_1(z,\lambda)$ and 
	$$Z_f=	\{z\in (\C^{\star})^d: f(z)=0\} . $$
	Let
	\[\tilde{f}(z)=f(z_1^{q_1},z_2^{q_2},\cdots,z_d^{q_d}).\]
	Solving  the equation $\det (A+B)=0$  and by \eqref{A},  we   have that 	if  $z_1=z_0^2$, $z_2=z_3=\cdots=z_{d-1}=z_0$ and $z_0\to 0$, then  $z_{d}\to 0$ or $z_d^{-1}\to 0$.  This implies that letting  $z_1=z_0^2$, $z_2=z_3=\cdots=z_{d-1}=z_0$ and $z_0\to 0$,  and solving the equation  $f(z)=0$, we must have either $z_{d}\to 0$ or $z_d^{-1}\to 0$. 
	Therefore,  the closure  of 
	$Z_f$   contains either 
	$z_1=z_2=\cdots=z_d=0$  or $z_1=z_2=\cdots=z_{d-1}=0,z_{d}^{-1}=0$.
	
	Take $z_1=z_2=\cdots=z_d=0$  into consideration first.
	Let $A$ and $B$ be given by Lemma \ref{lesep}.
	Then the  off-diagonal entries  of  $-z_{1} z_{2} \cdots z_d(A+B)$  are all divisible by $z_{1} z_{2}\cdots z_d,$ while the diagonal entries are
	\begin{equation}\label{hom2}
		\left(z_1z_2\cdots z_d\left(\sum_{j=1}^d\frac{1}{\rho^j_{n_j}z_j}\right)+\text { functions divisible by } z_{1} z_{2} \cdots z_d\right),
	\end{equation}
	where $0\leq n_j\leq q_j-1$. This shows  the  component  of   lowest degree of 
	$\det(-z_{1} z_{2} \cdots z_d(A+B))$  with respect to variables $z_1,z_2,\cdots,z_d$,   is
	\begin{equation}\label{hom1}
		\tilde{h}_1(z)=	z_1^Qz_2^Q\cdots z_d^Q\prod_{ 0\leq n_j\leq q_j-1\atop {1\leq j\leq q}} \left(\sum_{j=1}^d\frac{1}{\rho^j_{n_j}z_j}\right).
	\end{equation}
	{\bf Claim 1:}	by the fact that    $h_1(z)$  is  irreducible 	by Lemma \ref{key}, one has that there exists at most one factor $f(z)$ of $\mathcal{P}_1(z,\lambda)$ such that 
	the closure of $	\{z\in (\C^{\star})^d: f(z)=0\} $ contains  	$z_1=z_2=\cdots=z_d=0$.
	Claim 1  immediately follows from  some basic facts of algebraic geometry.  For convenience, we include an elementary  proof in the Appendix.

	
	Similarly, 
	the   component of  lowest degree of 
	$\det(-z_{1} z_{2} \cdots z_{d-1}z_d^{-1}(A+B))$  with respect to variables  $z_1,z_2,\cdots,z_{d-1},z_d^{-1}$
	is
	\begin{equation}\label{hom12}
		\tilde{h}_2(z)=	z_1^Qz_2^Q\cdots z_{d-1}^Qz_d^{-Q}\prod_{ 0\leq n_j\leq q_j-1\atop {1\leq j\leq q}} \left(\rho^d_{n_d} z_d+ \sum_{j=1}^{d-1}\frac{1}{\rho^j_{n_j}z_j}\right).
	\end{equation}	
	Since    $h_2(z)$ is irreducible  by Lemma \ref{key},   by a similar argument of the proof of Claim 1,    one has that there exists at most one factor $f(z)$ of $\mathcal{P}_1(z,\lambda)$ such that 
	the closure of $	\{z\in (\C^{\star})^d: f(z)=0\} $ contains   $z_1=z_2=\cdots=z_{d-1}=0,z_d^{-1}=0$.
	Therefore, $\mathcal{P}_{1}(z,\lambda)$ has at most two non-trivial factors. When $\mathcal{P}_{1}(z,\lambda)$ actually   has two factors, by the above analysis, the statements in Lemma \ref{lesin} hold.
\end{proof}
\begin{remark}
	When $d=2$, Gieseker, Kn\"orrer and Trubowitz proved that the Fermi variety $F_{\lambda}(V)/\Z^2$ has at most two irreducible  components for any $\lambda$ ~\cite[Corollary 4.1]{GKTBook}. 
	Even for $d=2$,
	our approach is different.  We   show that the closure of the zero set of every factor of $\mathcal{P}_1$ must contain either   $z_1=z_2=\cdots=z_d=0$ or  $z_1=z_2=\cdots=z_{d-1}=z_d^{-1}=0$ by solving algebraic equations on  properly choosing   curves.
\end{remark}
We are ready to prove Theorems \ref{thm2} and \ref{gcf}.
\begin{proof}[\bf Proof of Theorem \ref{thm2}]
	Without loss of generality, assume $[V]=0$.  
	Assume
	$\mathcal{P}(z,\lambda)$  is   reducible for some $\lambda\in \C$. By Lemma \ref{lesin},
	there are two non-constant polynomials $f(z)$ and $g(z)$ such that   none of them has a factor $z_1$ or $z_2$ (by \eqref{nomo}), and
	\begin{equation}
		\mathcal{P}_1(z,\lambda)=	(-1)^{q_1q_2} z_1^{q_2}z_2^{q_1} \mathcal{P}(z_1,z_2,\lambda) =f(z_1,z_2)g(z_1,z_2).
	\end{equation}
	Moreover, the closure of  $\{z\in (\C^{\star})^2: f(z)=0\} $ contains $z_1=z_2=0$  and the closure of 
	$\{z\in  (\C^{\star})^2: g(z)=0\}$ contains $z_1=0,z_2^{-1}=0$.
	
	Let
	\[\tilde{f}(z)=\tilde{f}(z_1,z_2)=f(z_1^{q_1},z_2^{q_2}), \tilde{g}(z)=\tilde{g}(z_1,z_2)=g(z_1^{q_1},z_2^{q_2}) .\]
	Therefore,   $\tilde{f}(z)$ and $\tilde{g}(z)$  are also  polynomials and  
	\begin{equation}\label{fg}
		\tilde{f}(z)\tilde{g}(z)=(-1)^{q_1q_2}z_1^{q_1q_2}z_2^{q_1q_2} \tilde{\mathcal{P}}(z_1,z_2,\lambda)={\rm det}(-z_1z_2A-z_1z_2B).
	\end{equation}
	
	By   \eqref{hom1} and \eqref{hom12},  we have  there exists a non-zero constant $K$ such that 
	\begin{equation}\label{f}
		\tilde{f}(z)=\left(\sum_{i=1}^{p}c_iz_1^{a_i}z_2^{b_i}\right)+K\prod_{ 0\leq n_1\leq q_1-1\atop {0\leq n_2\leq q_2-1}} \left(\frac{z_2}{\rho^1_{n_1}}+\frac{z_1}{\rho^2_{n_2}}\right),
	\end{equation}
	where $a_i+b_i\geq q_1q_2+1$, and
	
	\begin{equation}\label{g}
		\tilde{g}(z)=z_2^{k}\left[\left(\sum_{i=1}^{\tilde{p}}\tilde{c}_iz_1^{\tilde{a}_i}z_2^{-\tilde{b}_i}\right)+\prod_{ 0\leq n_1\leq q_1-1\atop {0\leq n_2\leq q_2-1}} \left(\frac{1}{z_2\rho^1_{n_1}}+ z_1\rho^2_{n_2} \right)\right],
	\end{equation}
	where $\tilde{a}_i+\tilde{b}_i\geq q_1q_2+1$ and  $k= \max_{1\leq i\leq \tilde{p}}\{q_1q_2,\tilde{b}_i\}$  (this ensures that $g(z)$ is a polynomial and $g(z)$ does not have a factor $z_2$).
	
	The matrix $z_1z_2A$ is given by $$ -\left(\rho^1_{n_1}z^2_1z_2+\frac{z_2}{\rho^1_{n_1} } +\frac{z_1}{ \rho^2_{n_2}}+\rho^2_{n_2} z_2 ^2 z_1 +\lambda z_1z_2\right)\delta_{n_1,n_1^{\prime}}\delta_{n_2,n_2^{\prime}}$$
	and all the entries of $z_1z_2 B$ only have a factor $z_1z_2$. Therefore,  by \eqref{fg}, 
	\begin{equation}\label{g14}
		\deg(\tilde{f})+\deg(\tilde{g})=\deg(\tilde{f}\tilde{g})=\deg(\det(-z_1z_2A-z_1z_2B)) \leq 3 q_1 q_2.
	\end{equation}
	By \eqref{f},  one has if $c_i=0$, $i=1,2,\cdots p$,
	\begin{equation}\label{g15}
		\deg (\tilde{f}) =q_1q_2,
	\end{equation}
	and 
	if one of  $c_i$, $i=1,2,\cdots p$, is nonzero, 
	\begin{equation}\label{g16}
		\deg (\tilde{f}) \geq q_1q_2+1.
	\end{equation}
	By \eqref{g}, one has
	\begin{equation}\label{g17}
		\deg (\tilde{g}) \geq k+q_1q_2.
	\end{equation}
	By \eqref{g14}-\eqref{g17} and the fact that $k= \max_{1\leq i\leq \tilde{p}}\{q_1q_2,\tilde{b}_i\}\geq q_1q_2$,  we must have
	$k=q_1q_2$, $ \tilde{b}_i\leq q_1q_2$ and   $c_i=0$, $i=1,2,\cdots, p.$ 
	Therefore, 
	\begin{equation}\label{f1}
		\tilde{f}(z)=K\prod_{ 0\leq n_1\leq q_1-1\atop {0\leq n_2\leq q_2-1}} \left(\frac{z_2}{\rho^1_{n_1}}+\frac{z_1}{\rho^2_{n_2}}\right).
	\end{equation}
	Reformulate \eqref{fg}, \eqref{g} and   \eqref{f1}   as,
	\begin{equation*}
		\frac{1}{z_2^{2q_1q_2}}\tilde{f}(z)\tilde{g}(z)=(-1)^{q_1q_2}{\rm det}\left[ \frac{z_1}{z_2} (A+B)\right],
	\end{equation*}
	\begin{equation*}
		\frac{1}{z_2^{q_1q_2}}\tilde{f}(z)=K\prod_{ 0\leq n_1\leq q_1-1\atop {0\leq n_2\leq q_2-1}} \left(\frac{1}{\rho^1_{n_1}}+\frac{z_1}{z_2\rho^2_{n_2}}\right),
	\end{equation*}
	and
	\begin{equation*}
		\frac{1}{z_2^{q_1q_2}}\tilde{g}(z)= \left[\left(\sum_{i=1}^{\tilde{p}}\tilde{c}_iz_1^{\tilde{a}_i}z_2^{-\tilde{b}_i}\right)+\prod_{ 0\leq n_1\leq q_1-1\atop {0\leq n_2\leq q_2-1}} \left(\frac{1}{z_2\rho^1_{n_1}}+ \rho^2_{n_2}z_1  \right)\right],
	\end{equation*}
	where $\tilde{a}_i+\tilde{b}_i\geq q_1q_2+1$ and $\tilde{b}_i\leq q_1q_2$.
	
	The matrix $\frac{z_1}{z_2}A$ is  $$ -\left(\rho^1_{n_1}\frac{z^2_1}{z_2}+\frac{1}{z_2\rho^1_{n_1} } + \frac{z_1}{\rho^2_{n_2}z_2^2}+  \rho^2_{n_2} z_1  +\lambda \frac{z_1}{z_2}\right)\delta_{n_1,n_1^{\prime}}\delta_{n_2,n_2^{\prime}}$$
	and every entry  of $\frac{z_1}{z_2} B$ only has a factor $\frac{z_1}{z_2}$.
	

	Since $z_1^{\tilde{a}_i}z_2^{-\tilde{b}_i}\prod_{ 0\leq n_1\leq q_1-1\atop {0\leq n_2\leq q_2-1}} \left(\frac{1}{\rho^1_{n_1}}+\frac{z_1}{z_2\rho^2_{n_2}}\right)$  with  $\tilde{a}_i+\tilde{b}_i\geq q_1q_2+1$  will contribute to $z_1^{i}z_2^{-j}$ with $i+j\geq 3q_1q_2+1$ and ${\rm det }( \frac{z_1}{z_2} (A+B))$ can only have $z_1^{\tilde{i}}z_2^{-\tilde{j}}$ with $\tilde{i}+\tilde{j}\leq 3q_1q_2$, a degree argument (regard $z_2^{-1} $ as a new variable) leads to  $\tilde{c}_i=0$, $i=1,2,\cdots, \tilde{p}  $.   
	Therefore,
	\begin{equation}\label{g1}
		\tilde{g}(z)= \prod_{ 0\leq n_1\leq q_1-1\atop {0\leq n_2\leq q_2-1}} \left(\frac{1}{\rho^1_{n_1}}+ \rho^2_{n_2}z_1 z_2 \right).
	\end{equation}
	We conclude that we  prove that  if $\mathcal{P}_1(z,\lambda)$ is reducible,  then  by \eqref{fg}, \eqref{f1} and \eqref{g1},
	there exists a constant $K\neq 0$ such that  
	\begin{align} 
		\det (-A&-B)\nonumber\\
		&=\frac{K}{z_1^{q_1q_2}z_2^{q_1q_2}}\prod_{ 0\leq n_1\leq q_1-1\atop {0\leq n_2\leq q_2-1}} \left(\frac{z_2}{\rho^1_{n_1}}+\frac{z_1}{\rho^2_{n_2}}\right) \prod_{ 0\leq n_1\leq q_1-1\atop {0\leq n_2\leq q_2-1}} \left(\frac{1}{\rho^1_{n_1}}+ \rho^2_{n_2}z_1 z_2 \right).	\label{last}
	\end{align} 
	We will prove that if \eqref{last} holds, then  $\lambda=0$.

	Let 
	\begin{align*}
		t_{n_1,n_2}(z_1,z_2)&
		= \rho^1_{n_1}z_1+\frac{1}{\rho^1_{n_1} z_1} +\rho^2_{n_2}z_2+\frac{1}{\rho^2_{n_2}  z_2} \\
		&=\left(\rho^1_{n_1}z_1 +\rho^2_{n_2}z_2\right) \left(1+\frac{1}{\rho^1_{n_1}\rho^2_{n_2}  z_1 z_2}\right).
	\end{align*}
	Then $ 	t_{n_1,n_2}(z_1,z_2)+\lambda$ is  the $(n_1,n_2)$-th diagonal entry of $A$.
	
	Let   $z_1=-z_2$.
	By \eqref{last},     one has
	\begin{equation}\label{last5}
		\det (A+B)\equiv 0.
	\end{equation}
	and 
	\begin{equation}\label{last7}
		t_{0,0}(z_1,z_2)\equiv 0.
	\end{equation}
	Since $q_1$ and $q_2$ are coprime, for any $(n_1,n_2)\in W\backslash (0,0)$, 
	\begin{equation}\label{last8}
		\rho^1_{n_1}z_1 -\rho^2_{n_2}z_1\neq 0, \text{ for } z_1\neq 0,
	\end{equation}
	and hence
	$t_{n_1,n_2}$ is not a zero function.
	Check the term of highest  degree of $z_1$($z_2$) in  $\det (A+B)$. 
	By \eqref{A}, \eqref{last7} and \eqref{last8},  the term of  highest  degree  (up to a nonzero constant factor) is 
	\begin{equation}\label{ggequ7}
		\lambda z_1^{q_1q_2-1}.
	\end{equation}
	By \eqref{last5} and \eqref{ggequ7},
	$\lambda=0$. We complete the proof of the  first part of Theorem \ref{thm2}. The second part  follows from \eqref{last}.
\end{proof}

\begin{proof} [\bf Proof of Theorem \ref{gcf}] 
	The proof is similar to that of Theorem \ref{thm2}. 
	Without loss of generality, assume $[V]=0$.  
	Assume that 	$\mathcal{P}(z,\lambda)$  is   reducible.  
	Then there are two non-constant polynomials $f(z)$ and $g(z)$ such that   none of them has a factor $z_j$, $j=1,2,\cdots,Q$, and 
	\begin{equation}
		(-1)^{Q} z_1^{\frac{Q}{q_1} }z_2^{\frac{Q}{q_2}}\cdots z_d^{\frac{Q}{q_d}} \mathcal{P}(z,\lambda) =f(z)g(z).
	\end{equation}
	Let 
	\[\tilde{f}(z)=f(z_1^{q_1},z_2^{q_2},\cdots,z_d^{q_d}), \tilde{g}(z)=g(z_1^{q_1},z_2^{q_2},\cdots,z_d^{q_d}) .\]
	Therefore,   $\tilde{f}(z)$ and $\tilde{g}(z)$  are also  polynomials and  
	\begin{align}
		\tilde{f}(z)\tilde{g}(z)=&(-1)^{Q}z_1^{Q}z_2^{Q} \cdots z_d^Q\tilde{\mathcal{P}}(z,\lambda)\nonumber\\
		=&{\rm det} (-z_{1} z_{2} \cdots z_d(A+B)). \label{fgnew}
	\end{align}
	Moreover, the closure  of  $\{z\in (\C^{\star})^d: f(z)=0\}$ contains $z_1=z_2=\cdots=z_d= 0 $ and the closure of 
	$ \{z\in (\C^{\star})^d: g(z)=0\}$ contains  $z_1=z_2=\cdots=z_{d-1}= 0 $ and $z_d^{-1}=0$.
	
	By   \eqref{hom1} and \eqref{hom12},  we have for some non-zero constant $K$,
	\begin{equation}\label{fnew}
		\tilde{f}(z)=\left(\sum_{i=1}^{p}c_iz^{a_i} \right)+K\tilde{h}_1(z),
	\end{equation}
	where $||a_i||_1\geq (d-1)Q+1$, and
	
	\begin{equation}\label{gnew}
		\tilde{g}(z)=z_d^{k}\left[\left(\sum_{i=1}^{\tilde{p}}\tilde{c}_i\tilde{z}^{\tilde{a}_i}z_d^{-\tilde{b}_i}\right)+\tilde{h}_2(z) \right],
	\end{equation}
	where   $\tilde{z}=(z_1,z_2,\cdots,z_{d-1})$,  $||\tilde{a}_i||_1+\tilde{b}_i\geq (d-1)Q+1$ and  $k= \max_{1\leq i\leq \tilde{p}}\{Q,\tilde{b}_i\}$.
	
	By \eqref{fnew}, one has
	\begin{equation}\label{g31}
		\deg(\tilde{f})\geq \deg(\tilde{h}_1)=(d-1)Q.
	\end{equation}
	By \eqref{gnew},
	\begin{equation}\label{g32}
		\deg(\tilde{g})\geq \deg(z_d^k\tilde{h}_2(z))\geq \deg(z_d^Q\tilde{h}_2(z))=dQ.
	\end{equation}
	By \eqref{g31}, \eqref{g32} and \eqref{fgnew}, one has
	\begin{equation*}
		\deg(\det(z_1z_2\cdots z_d (A+B)) )=\deg(\tilde{f}\tilde{g})\geq (2d-1)Q.
	\end{equation*}
	This is impossible since $\deg({\rm det} (z_1z_2\cdots z_d(A+B)))\leq (d+1)Q$.
\end{proof}

\begin{proof}[\bf Proof of Theorem \ref{thmbv1ne}]
	Assume  $\mathcal{P}(z,\lambda) $ is irreducible. Then there exist two non-trivial factors $f_j(z,\lambda)$, Laurent polynomial in $z$ and polynomial in $\lambda$, $j=1,2$, such that $\mathcal{P}(z,\lambda)=f_1(z,\lambda)f_2(z,\lambda)$.  Rewrite $f_j(z,\lambda)$, $j=1,2$, as
	\begin{equation*}
		f_j(z,\lambda)= \sum_{a\in A_j} t_j^a(\lambda) z^a,
	\end{equation*}
	where $t_j^a(\lambda)$ is a polynomial of $\lambda$ and $A_j$ is a proper finite subset of $\Z^d$. Let $\lambda$ be large enough so that 
	for any $j=1,2$ and $a\in A_j$, $t_j^a(\lambda)\neq 0$.
	
	By Theorems \ref{gcf1} and \ref{thm21},  $\mathcal{P}(z,\lambda)$ (as  a function of  variables $z$) is irreducible  for any large enough $\lambda$. Therefore, we must have that for any  large enough $\lambda$,  either $f_1(z,\lambda) $ or $f_2(z,\lambda)$ is a monomial of $z$. Then we conclude that either the cardinality of $A_1$ is one or the cardinality of $A_2$  is one.
	Without loss of generality assume  that $f_1(z,\lambda)=t^{a_0}_1 (\lambda)z^{a_0}$ for some $a_0\in \Z^d$.   Since $f_1$ is non-monomial, one has that $t^{a_0}_1(\lambda)$ is non-constant. 
	Let $\lambda_0\in\C$ be such that  $t^{a_0}_1(\lambda_0)=0$.  Then we have $\mathcal{P}(z,\lambda_0)=0$ for any $z$. 
	Recall that the highest degree term (up to a $\pm$ sign) of $z_1$ in  $\mathcal{P}(z,\lambda_0)$ is $z_1^{\frac{Q}{q_1}}$
	(Fact (2) at the end of Section 4).  We obtain the contradiction. 
	
\end{proof}

	\section{Proof of Theorem \ref{thm1} }\label{Semb}

	\begin{theorem}~\cite[Lemma 17]{kvcpde20}\label{thmentire}
		\label{divis}Let $Z$ be the set of all zeros of an entire function $\zeta (k)
		$ in $\C^d$ and $ \cup Z_j$ be its irreducible components. Assume that the
		real part $Z_{j,\R}=Z_j\cap \R^d$ of each $ Z_j$ contains a
		submanifold of real dimension $d-1$. Let also $g(k)$ be an entire function
		in $\C^d$ with values in a Hilbert space $\mathcal{H}$ such that on the real
		space $\R^d$ the ratio 
		\[
		f(k)=\frac{g(k)}{\zeta (k)}
		\]
		belongs to $L^2_{loc}(\R^d,\mathcal{H})$. Then $f(k)$ extends to an
		entire function with values in $\mathcal{H}$.

	\end{theorem}
	The following lemma is well known,  we include a proof here for completeness.
	
	\begin{lemma}\label{leande}
		Let  $\hat{f}\in L^2(\T^d)$ and $\{{f}_n\}$ be its Fourier series, namely, for $n\in \Z^d$,
		\begin{equation*}
		f_n=\int_{\T^d} \hat{f}(x)e^{-2\pi i n\cdot x} dx.
		\end{equation*}
		Then the following statements are true:
		\begin{itemize}
			\item [i). ]If $\hat{f} $ is an entire function and $|\hat{f}(z)|\leq Ce^{C|z|^{r}}$ for some $C>0$ and $r>1$, then for any $0<w<\frac{r}{r-1} $,
			\begin{equation*}
			|{f}_n|\leq e^{-  |n|^{w}},
			\end{equation*}
			for large enough $n$.
			\item [ii).] If $	|{f}_n|\leq Ce^{-C^{-1} |n|^{r}} $ for some $C>0$ and $r>1$,  then $\hat{f}$ is an entire function and  there exists a constant  $C_1$ (depending on $C$ and dimension $d$) such that 
			\begin{equation*}
			|\hat{f}(z)|\leq e^{C_1 |z|^{\frac{r}{r-1}}},
			\end{equation*}
			for large enough $|z|$.
		\end{itemize}
		
	\end{lemma}
	\begin{proof}
		
		Fix any large $n=(n_1,n_2,\cdots,n_d)\in \Z^d$. Without loss of generality, assume $n_1>0$ and $n_1=\max\{|n_1|,|n_2|,\cdots,|n_d|\}$. Then for any $\tilde{w}<\frac{1}{r-1}$,
		\begin{align*} 
		|f_n |&= \left|\int_{\T^d}\hat{f}(x)e^{-2\pi i n \cdot x}dx\right| \\ 
		&=   \left|\int_{\T^{d-1}} e^{-2\pi i (n_2  x_2+\cdots n_dx_d )} dx_2\cdots dx_d\int_{z_1=x-in_1^{\tilde{w}}\atop {x\in \T}}\hat{f}(z)e^{-2\pi i n_1 z_1}dz_1\right|\\
		&\leq Ce^{C n_1^{r\tilde{w}}}e^{-2\pi n_1^{1+\tilde{w}}}\\
		&\leq  e^{-  n_1^{1+\tilde{w}}},
		\end{align*}
		for large $|n|$. This proves i). 
		
		Obviously,
		\begin{equation*}
		\hat{f}(z)=\sum_{n\in \Z^d}f_n e^{2\pi i n\cdot z}.
		\end{equation*}
		Then one has
		\begin{align*} 
		|\hat{f}(z) |&\leq\sum_{n\in \Z^d} Ce^{-C^{-1}|n|^r} e^{C|n||z|}\\ 
		&\leq\sum_{l=1}^{\infty} C l^de^{-C^{-1}l^r} e^{Cl|z|}\\ 
		&\leq  e^{C |z|^{\frac{r}{r-1} }},
		\end{align*}
		for any large $z$. This completes the proof of ii).
	\end{proof}

	\begin{lemma} \label{thmgrowth}
		Let  $f$ and $g $ be entire functions on $\C^d$.
		Assume that  for some $ C_1>0,\rho>0$, 
		\begin{equation}\label{ggrowth}
		|f(z)|\leq C_1e^{C_1|z|^{\rho}}, |g(z)|\leq C_1e^{C_1|z|^{\rho}}.
		\end{equation}
		Assume that  $h=g/f$ is also an entire function on $\C^d$. Then there exists a constant $C$ such that 
		\begin{equation*}
		|h(z)|\leq C e^{ C |z|^{\rho}}.
		\end{equation*} 
	\end{lemma}
	
	\begin{remark}
		Lemma \ref{thmgrowth} is well known, e.g.,  see   Theorem 5 of Section 11.3 in  ~\cite{lev}  for $d=1$  and  p.37 in ~\cite{kuflo} for $d\geq 2$. 
	\end{remark}

	The following Lemma can be obtained by a   straightforward computation. 
	For example, see   p.49 in  Bourgain-Klein ~\cite{bk13} or  Lyubarskii-Malinnikova  ~\cite{LM18}.
	\begin{lemma} \label{lecontinuation}
		Let $\tilde{V}:\Z^d\to \C$ be bounded. Assume that $u$ is  a non-trivial solution of 
		\begin{equation*}
		(-\Delta+\tilde{V})u=0.
		\end{equation*}
		Then  for some constant $C>0$,
		\begin{equation*}
		\sup_{|n|=R}	(|u(n)|+|u(n-1)|)\geq e^{-C R}.
		\end{equation*}
	\end{lemma}
We are ready to prove Theorem \ref{thm1}.
	\begin{proof}[\bf Proof of Theorem \ref{thm1}]
		
				Suppose  there exists  $\lambda\in(a_m,b_m) $ such that  $\lambda\in \sigma_{p}(H)$. Then there
		exists a non-zero function $u\in \ell^2(\Z^d)$ such
		that 
		\[
		-\Delta u+Vu+vu=\lambda u, 
		\]
		or 
		
		\begin{equation}\label{geigen1}
		(H_0-\lambda I)u=-vu.
		\end{equation}
		Denote by the function on the right hand side by $\psi (n):$%
		\[\psi (n)=-v(n)u(n),n\in\Z^d. \]
		Applying $U$ on both sides of \eqref{geigen1}, one has
		\begin{equation}
	(	(\hat{H}_0-\lambda I)\hat{u})(x,l)=\hat{\psi}(x,l),
		\end{equation}
		where $\hat{u}(x,l)\in L^2(\mathcal{B}\times \bar{W})$. 
		For any fixed $x$, we regard both $\hat{u}(x,\cdot)$ and  $\hat{\psi}(x,\cdot)$  as vectors on $\bar{W}$. 
		Therefore, for any $x\in \mathcal{B}$, 
			\begin{equation}
			(\tilde{H}_0(x)-\lambda I)\hat{u}(x,\cdot)=\hat{\psi}(x,\cdot).
		\end{equation}

		By the assumption \eqref{ggdecay} and  Lemma \ref{leande}, we have that for any $l\in  \bar{W}$, 
		\begin{equation}\label{gpsi}
		|\hat{\psi}(x,l)|\leq Ce^{C |x|^{\frac{\gamma}{\gamma-1}}}.
		\end{equation}
		From  Lemma \ref{le27} ($\tilde{H}_0(x)$ is unitarily equivalent to $\tilde{D}(x)$),  one can see that  $\det (\tilde{H}_0(x)-\lambda I)= \tilde{P}(x,\lambda)$. 
		By the Cramer's rule,  we have 
		\begin{equation*}
		(\tilde{H}_0(x)-\lambda I)^{-1}=\frac{\tilde{S}(x,\lambda)}{\tilde{P} (x,\lambda)}, 
		\end{equation*}
	 where $ \tilde{S}(x,\lambda)$ is  the adjoint matrix of $\tilde{H}_0(x)-\lambda I$.
		This concludes that
		\begin{equation*}
		\hat{u}(x,\cdot)=\frac{\tilde{S}(x,\lambda)\hat{\psi}(x,\cdot)}{\tilde{P} (x,\lambda)}.
		\end{equation*}
	When $\lambda$ satisfies Condition 1, by \eqref{g11}, one can see that $\zeta(x)=\tilde{P} (x,\lambda)$  satisfies the assumption of Theorem \ref{thmentire}.
		Since  $\hat{u}(x,l)\in L^2(\mathcal{B}\times \bar{W})$, namely  for any fixed $l\in\bar{W}$,  $\hat{u}(x,l)\in L^2(\mathcal{B})$, by Theorem \ref{thmentire},  one has that $\hat{u}(x,l)$ is an entire function in the variable $x$ for any $l\in \bar{W} $.
		Since all non-constant  entries  (in variables $x$) of $\tilde{H}_0(x)-\lambda I$ are consisted of   $e^{2\pi i  x_j}$  and  $e^{-2\pi i x_j}$,
		we have that 
		\begin{equation}\label{bpg}
		||\tilde{S}(x,\lambda)||\leq C e^{ C|x|}, |\tilde{P}(x,\lambda)|\leq C e^{C|x|}.
		\end{equation}
		By \eqref{gpsi} and \eqref{bpg}, 
		one has that   $ \tilde{P} (x,\lambda)$  	satisfies  \eqref{ggrowth} with $\rho=1$ and   for  any $l\in \bar{W}$,  $(\tilde{S}  \hat{\psi})(x,l)$  	satisfies  \eqref{ggrowth} with   $\rho=\frac{\gamma}{\gamma-1}$.  By Lemma \ref{thmgrowth}, we have  that   for any $l\in \bar{W}$,
		\begin{equation*}
		|\hat{u}(x,l)|\leq C e^{C|x|^{\frac{\gamma}{\gamma-1}}}.
		\end{equation*}
		By Lemma \ref{leande}, we have that for any $w$ with $w<\gamma$,
		\begin{equation*}
		|u(n)|\leq Ce^{-|n|^{w}}.
		\end{equation*}
		This is contradicted to Lemma \ref{lecontinuation}.

	\end{proof}
	\section{ Proof of Theorem  \ref{thmextrem}}\label{sapp}
	
\begin{proof}[\bf Proof of Theorem \ref{thmextrem}]
	Clearly,  $(k,\lambda=\lambda_j(k))$, $j=1,2,\cdots,Q$,  is one branch of solutions of  equation 
	\begin{equation}\label{gc}
		{P}(k,\lambda)=\mathcal{P}(e^{2\pi ik_1},e^{2\pi ik_2},\cdots,e^{2 \pi i k_d},\lambda)=0,
	\end{equation}
	and
	\begin{equation}\label{gc3}
		{P}(k,\lambda)=\prod_{j=1}^Q(\lambda_j(k) -\lambda).
	\end{equation}
	Assume that $k_0=(k_0^1,k_0^2,\cdots,k_0^d)$ satisfies $\lambda_m(k_0)=\lambda_*$.
	Considering the matrix $D(k_0)$,
	let $m_1\geq 1$ be the   multiplicity  of  its  eigenvalue  $\lambda_*$.
	
	{\bf Case 1}: $m_1=1$.
	
	It means $\lambda=\lambda_*$ is a single root of ${P}(k_0,\lambda)=0$. Then $\partial_{\lambda} {P}(k_0,\lambda)|_{\lambda=\lambda_*}\neq 0$.
	By the implicit function theorem,  $\lambda_m(k)$ is an analytic function in a neighborhood of $k_0$.
	Since $\lambda_{*}=\lambda_{m}(k_0)$ is an extremum,  one has  
	\begin{equation}\label{gc1}
		\nabla_k \lambda_m(k)|_{k=k_0}=(0,0,\cdots,0).
	\end{equation}
	Rewrite  \eqref{gc3} as
	\begin{equation}\label{gc10}
		{P}(k, \lambda_*)=(\lambda_m(k)-\lambda_{*}) T(k),
	\end{equation}
	where $T(k)$ is analytic in  a neighborhood of $k_0$.
	By \eqref{gc1} and \eqref{gc10}, we have 
	
	\begin{equation}\label{gc14}
		\nabla_k {P}(k,\lambda_{*})|_{k=k_0}=(0,0,\cdots,0).
	\end{equation}
	
	{\bf Case 2}:  $m_1\geq 2$.
	
	We will  show that \eqref{gc14} still holds  in this case. Without loss of generality, we only prove that 
	\begin{equation}\label{gc11}
		\partial_{k_1} {P}(k,\lambda_{*})|_{k=k_0}=0.
	\end{equation}
	In order to prove \eqref{gc11}, it suffices to show that
	\begin{equation}\label{gc12}
		\partial_{k_1} {P}(k_1,k_0^2,\cdots,k_0^d,\lambda_{*})|_{k_1=k_0^1}=0.
	\end{equation} 
	By the Kato-Rellich perturbation theory ~\cite{katoan}, there exists $\tilde{\lambda}_l (k_1)$,  $l=1,2, \cdots,m_1$,  such that 
	in a neighborhood of $k_0^1$,  $\tilde{\lambda}_l (k_1)$ is analytic, $\tilde{\lambda}_l (k_0^1)=\lambda_{*}$ and $\tilde{\lambda}_l (k_1)$ is  an eigenvalue of $D(k_1,k_0^2,\cdots,k_0^d)$,  $l=1,2, \cdots,m_1$.
	Moreover,
	\begin{equation}\label{gc13}
		{P}(k_1,k_2^0, \cdots,k_d^0,\lambda_*)=T(k_1)\prod_{l=1}^{m_1}(\tilde{\lambda}_{l}(k_1) -\lambda_*),
	\end{equation}
	where $T(k_1)$ is analytic  in a neighborhood of $k_0^1$.
	Now \eqref{gc12} follows from \eqref{gc13}.   We complete the proof.
\end{proof}

	\section{ Proof of Theorems    \ref{thmex2}, \ref{thmex3}	and \ref{thm1new}}\label{sapp1}

 \begin{proof}[\bf Proof of Theorem 
	\ref{thmex2}]
	
	By Lemma \ref{lesin}, the polynomial $z_1^{q_2}z_2^{q_1}\mathcal{P}(z,\lambda)$ (as a function of $z_1$ and $z_2$) is  square-free
	for any $\lambda$.  By B\'ezout's theorem,
	we have that
	\begin{equation*} 
		\# \{z\in (\C^\star)^2:\mathcal{P}(z,\lambda_{*})=0,|\nabla_z\mathcal{P}(z,\lambda_{*})|=0\}\leq 4(q_1+q_2)^2,
	\end{equation*}
	and hence
	\begin{equation}\label{last9}
		\# \{k\in [0,1)^2:{P}(k,\lambda_{*})=0,|\nabla_k{P}(k,\lambda_{*})|=0\}\leq 4(q_1+q_2)^2,
	\end{equation}
	Now  Theorem 
	\ref{thmex2}  follows from   \eqref{gextrem} and \eqref{last9}.

\end{proof}

\begin{proof}[\bf Proof of Theorem 
	\ref{thmex3}]
	By Lemma \ref{lesin},  $z_1^{\frac{Q}{q_1} }z_2^{\frac{Q}{q_2}}\cdots z_d^{\frac{Q}{q_d}}\mathcal{P}(z,\lambda_{*})$ is square-free, then   by the basic fact of analytic sets (e.g., Corollary 4 in p.69 ~\cite{nr66}),   the analytic set 
	$ \{z\in (\C^\star)^d:\mathcal{P}(z,\lambda_{*})=0,|\nabla_z\mathcal{P}(z,\lambda_{*})|=0\} $ has (complex) dimension at most $d-2$.  Since the real dimension of a real analytic set is always smaller than or equal to the complex dimension  (e.g.,  p.63 in ~\cite{nr66}), 
one has that  $\{k\in [0,1)^d:{P}(k,\lambda_{*})=0,|\nabla_k{P}(k,\lambda_{*})|=0\}$    has dimension  at most $d-2$. Now Theorem 
	\ref{thmex3} follows from   \eqref{gextrem}.
	
\end{proof}
\begin{remark}\label{relast}
	In the proof of Theorems 
	\ref{thmex2} and \ref{thmex3}, we only use the fact that the polynomial  $z_1^{\frac{Q}{q_1} }z_2^{\frac{Q}{q_2}}\cdots z_d^{\frac{Q}{q_d}}\mathcal{P}(z,\lambda_{*})$ (as a function of $z$) is square-free.
\end{remark}
\begin{lemma}~\cite[Lemma 4]{kv06cmp}\label{led-1}
	Let $d\geq 2$.
	Assume $\lambda\in (a_m,b_m)$ for some $m$. Then the Fermi variety $F_{\lambda}(V)$ contains an open  analytic hypersurface of  dimension $d-1$ in $\R^d$.
\end{lemma}

\begin{proof}[\bf Proof of Theorem 
	\ref{thm1new}]
	For $d=1$, $ H_0+v$ does not have embedded eigenvalues  if $v(n)=\frac{o(1)}{|n|}$ as $n\to \infty$~\cite{ld}.
	Therefore, it suffices to  prove Theorem \ref{thm1new} for $d\geq 2$.
	
	By Lemma  \ref{led-1},   if  $\lambda\in \cup(a_m,b_m)$ and ${F}_{\lambda}(V)$ is irreducible, then $\lambda$ satisfies   Condition 1.
	For $d=2$, if $F_{\lambda}(V)$ is reducible, by Theorem \ref{thm21}, $\lambda=[V]$. By \eqref{last}, $\lambda=[V]$ satisfies
	Condition 1. 
	For $d\geq 3$,  by Theorem \ref{gcf1}, the Condition 1 holds for every $\lambda\in \cup(a_m,b_m)$.
	Now Theorem 
	\ref{thm1new} follows from Theorem \ref{thm1}.
\end{proof}
\appendix
\section{ Proof of Claim 1  }

\begin{proof}
	Otherwise, $\mathcal{P}_1(z,\lambda)$ has two non-trivial polynomial factors $f(z)$ and $g(z)$ such that 
	both  $\{z\in\C^d: f(z)=0\}$ and   $\{z\in \C^d: g(z)=0\}$ contain $z_1=z_2=\cdots=z_d=0$.
	Let 
	\begin{equation*}
		\tilde{f}(z)=f(z_1^{q_1},z_2^{q_2},\cdots,z_d^{q_d}),   \tilde{g}(z)=g(z_1^{q_1},z_2^{q_2},\cdots,z_d^{q_d}).
	\end{equation*}
	Let $\tilde{f}_1(z)$ ($\tilde{g}_1(z)$)  be the  component of the  lowest degree of $\tilde{f}(z)$  ($\tilde{g}(z)$). 
	Since both  $\{z\in\C^d: f(z)=0\}$ and   $\{z\in \C^d: g(z)=0\}$ contain $z_1=z_2=\cdots=z_d=0$, one has  that 
	$\tilde{f}_1(z)$  and $\tilde{g}_1(z)$  are non-constant.

	Since both $\tilde{f}(z)$ and $\tilde{g}(z)$ are polynomials of $z_1^{q_1},z_2^{q_2},\cdots, z_d^{q_d}$, we have 
	$\tilde{f}_1(z)$  and $\tilde{g}_1(z)$ are also polynomials of $z_1^{q_1},z_2^{q_2},\cdots, z_d^{q_d}$ and hence there exist 
	$f_1(z)$ and $g_1(z)$ such that 
	\begin{equation*}
		\tilde{f}_1(z)=f_1(z_1^{q_1},z_2^{q_2},\cdots,z_d^{q_d}),   \tilde{g}_1(z)=g_1(z_1^{q_1},z_2^{q_2},\cdots,z_d^{q_d}).
	\end{equation*}
	By \eqref{hom2} and \eqref{hom1}, one has
	\begin{equation*}
		\tilde{f}_1(z) \tilde{g}_1(z)=\tilde{h}_1(z)
	\end{equation*}
	and hence
	\begin{equation*}
		{f}_1(z) {g}_1(z)= {h}_1(z).
	\end{equation*}
	This is impossible since $h_1(z)$ is irreducible.
	
\end{proof}

%
%
%
%

\section*{Acknowledgments}
I would like to thank  Constanza Rojas-Molina for drawing me attention to  ~\cite{kvcpde20}  
and the organizers of  the Workshop ``Spectral Theory of Quasi-Periodic and Random Operators"  in CRM, November 2018, during which this research was started. 
I    wish to  thank  Ilya Kachkovskiy and   Peter Kuchment    for comments on earlier versions of the manuscript, which 
greatly improved  the exposition.    I also wish  to thank  Rupert Frank and Simon Larson  for inviting me to give a talk  in the  ``38th Annual Western States
Mathematical Physics Meeting". During the meeting, Rupert Frank's comments 
made  me realize that   proofs of  the irreducibility  work for complex-valued  potentials without any changes. 
Finally, I wish to express my  gratitude to  anonymous referees, whose  comments  greatly helped the exposition of the manuscript.
This research was 
supported by   NSF DMS-1700314/2015683,  DMS-2000345 and DMS-2052572.

 \bibliographystyle{abbrv} 

\end{document}